\documentclass[10pt,letterpaper]{article} 

\usepackage{amsmath,amsfonts,amsthm,amssymb,dsfont}
\usepackage{comment}
\usepackage[linesnumbered,boxed,ruled,vlined]{algorithm2e}
\usepackage{url}
\usepackage{hyperref}
\usepackage{fullpage}
\usepackage{authblk}

\theoremstyle{plain}

\newtheorem{fram}{Framework}
\newtheorem{lemma}{Lemma}
\newtheorem{theo}{Theorem}[section]

\newtheorem{rem}[theo]{Remark}
\newtheorem{conjecture}{Conjecture}

\newtheorem*{reminderInternal}{Reminder: \reminderCurrent}
\newenvironment{reminder}[1]{\def\reminderCurrent{#1}\begin{reminderInternal}}{\end{reminderInternal}}
\theoremstyle{definition}

\everymath{\displaystyle}

\newcommand{\poly}{\operatorname*{poly}}

\newcommand{\PTIME}{\mathsf{P}}

\newcommand{\polylog}{\operatorname*{polylog}}

\renewcommand{\epsilon}{\varepsilon}

\def\ShowAuthNotes{1}
\ifnum\ShowAuthNotes=1
\newcommand{\authnote}[2]{\ \\ \textcolor{red}{\parbox{0.9\linewidth}{[{\footnotesize {\bf #1:} { {#2}}}]}}\newline}
\else
\newcommand{\authnote}[2]{}
\fi

\newcommand{\entry}{\mathcal{E}}
\newcommand{\partver}{\mathsf{part}}

\let\svfootnoterule\footnoterule
\renewcommand\footnoterule{\vfill\svfootnoterule}
\title{Nearly Optimal Separation Between Partially And Fully Retroactive Data Structures}
\date{}
\author[1]{Lijie Chen\thanks{Supported by an Akamai Fellowship.}}
\author[1]{Erik D. Demaine}
\author[1]{Yuzhou Gu}
\author[1]{Virginia Vassilevska Williams\thanks{Partially supported by an NSF Career Award, a Sloan Fellowship, NSF Grants CCF-1417238, CCF-1528078 and CCF-1514339, and BSF Grant BSF:2012338.}}
\author[1]{Yinzhan Xu}
\author[1]{Yuancheng Yu}
\affil[1]{MIT \texttt{\{lijieche,edemaine,yuzhougu,virgi,xyzhan,ycyu\}@mit.edu}}

\newcommand{\SAT}{\textsf{SAT}}
\newcommand{\CNFSAT}{\textsf{CNF\ SAT}}
\newcommand{\CircuitSAT}{\textsf{Circuit\ SAT}}
\newcommand{\SIZE}{\textsf{SIZE}}
\newcommand{\SETH}{\textsf{SETH}}
\newcommand{\highlight}[1]{\medskip \noindent \textbf{#1}}

\newenvironment{about}{\begin{list}{}{\leftmargin=2em\partopsep=0pt}\item}{\end{list}}

\begin{document}
	\maketitle
\begin{abstract}
          Since the introduction of retroactive data structures at SODA 2004, a major unsolved problem has been to bound the gap between the best partially retroactive data structure (where changes can be made to the past, but only the present can be queried) and the best fully retroactive data structure (where the past can also be queried) for any problem. It was proved in 2004 that any partially retroactive data structure with operation time $T_{\mathsf{op}}(n,m)$ can be transformed into a fully retroactive data structure with operation time $O(\sqrt{m} \cdot T_{\mathsf{op}}(n,m))$, where $n$ is the size of the data structure and $m$ is the number of operations in the timeline~\cite{demaine2004retroactive}. But it has been open for 14 years whether such a gap is necessary.

          In this paper, we prove nearly matching upper and lower bounds on this gap for all $n$ and $m$.  We improve the upper bound for $n \ll \sqrt m$ by showing a new transformation with multiplicative overhead $n \log m$.  We then prove a lower bound of $\min\{n \log m, \sqrt m\}^{1-o(1)}$ assuming any of the following conjectures:

          \begin{itemize}
              \item \textbf{Conjecture I:} $\CircuitSAT$ requires $2^{n - o(n)}$ time on $n$-input circuits of size $2^{o(n)}$. 
  %
  %

              \begin{about}
              This conjecture is far weaker than the well-believed $\SETH$ conjecture from complexity theory, which asserts that $\CNFSAT$ with $n$ variables and $O(n)$ clauses already requires $2^{n-o(n)}$ time.
              \end{about}

              \item \textbf{Conjecture II:} Online $(\min,+)$ product between an integer $n \times n$ matrix and $n$ vectors requires $n^{3 - o(1)}$ time.

              \begin{about}
              This conjecture is weaker than the $\mathsf{APSP}$ conjectures widely used in fine-grained complexity.
              \end{about}

              \item \textbf{Conjecture III (3-SUM Conjecture):} Given three sets $A,B,C$ of integers, each of size $n$, deciding whether there exist $a \in A, b \in B, c \in C$ such that $a + b + c = 0$ requires $n^{2 - o(1)}$ time.

              \begin{about}
              This 1995 conjecture \cite{Gajentaan-Overmars-1995} was the first conjecture in fine-grained complexity.
              \end{about}
          \end{itemize}


          Our lower bound construction illustrates an interesting power of fully retroactive queries: they can be used to quickly solve batched pair evaluation. We believe this technique can prove useful for other data structure lower bounds, especially dynamic ones.
		
	\end{abstract}
	
	\section{Introduction}
	
	\paragraph*{Retroactive Data Structures}
    A data structure can be thought of as a sequence of updates being applied to an initial state. In traditional data structures, we can only append updates to the end of this sequence, called the \emph{timeline}, and can only query about the final state of the data structure resulting from all the updates. \emph{Retroactive data structures}, introduced at SODA 2004~\cite{demaine2004retroactive}, allow us to add or remove updates in the past, i.e., anywhere in the timeline rather than only at the end.
	
	There are two main kinds of retroactive data structures: \emph{partially retroactive} data structures, where we are only allowed to query the present, i.e., the final version resulting from the whole update sequence; and \emph{fully retroactive} data structures, where we are also allowed to query about a past state, i.e., the state resulting from applying only a \emph{prefix} of the update sequence given by the timeline.

    Unlike persistence \cite{Driscoll-Sarnak-Sleator-Tarjan-1989}, there is no general efficient transformation from a data structure into a retroactive data structure, even partially retroactive with sublinear multiplicative overhead \cite{demaine2004retroactive}.
    Nonetheless, several efficient retroactive data structures have been developed
    \cite{demaine2007retroactive,Blelloch-2008,Giora-Kaplan-2009,Dickerson-Eppstein-Goodrich-2010,Nekrich-2010,Goodrich-Simons-2011,Parsa-2014,demaine2015polylogarithmic}.

	\paragraph*{Motivation: Full Retroactivity versus Partial Retroactivity}
    A key problem, posed in the original paper on retroactive data structures~\cite{demaine2004retroactive}, is whether the full retroactivity requirement makes problems much harder than their partially retroactive counterpart. The same paper established an $O(\sqrt{m})$ multiplicative overhead transformation from a partially retroactive data structure to a fully retroactive one, where $m$ is the number of updates in the timeline. 
	
		Prior to our work, there was no data structure problem whose best known fully retroactive version was substantially (more than a polylogarithmic factor) worse than the best known partially retroactive version. Priority queues \emph{used to be} the only problem with a polynomial gap (between $O(\sqrt m \log m)$ and $O(\log m)$ time~\cite{demaine2004retroactive}). But at WADS 2015 it was shown that priority queues have a polylogarithmic fully retroactive solution~\cite{demaine2015polylogarithmic}, and more generally, any ``time-fusable'' data structure can be transformed from partial to full retroactivity with polylogarithmic overhead.
       {\em Can this transformation be generalized to all data structures?}
	
	\paragraph*{Our Results: Conditional Lower Bounds}
    We show that, perhaps surprisingly, the $O(\sqrt{m})$ overhead for transforming partial retroactivity into full retroactivity is nearly optimal for general data structure problems, conditioned on any of three well-believed conjectures:
    
\begin{conjecture} \label{CircuitSAT}
    In the Word-RAM model of computation with $O(\log n)$ bit words, it takes $2^{n - o(n)}$ time to solve $\SIZE(2^{o(n)})$ $\CircuitSAT$: decide whether a given $n$-input circuit $C$ of size $2^{o(n)}$ is satisfiable.
\end{conjecture}
	
\begin{rem}
	The problem $\SIZE(2^{o(n)})$ $\CircuitSAT$ is far harder than $\CNFSAT$, and the conjecture above is much weaker than the well-believed Strong Exponential Time Hypothesis (\SETH) \cite{Impagliazzo2001SETH} which states that 
for every $\varepsilon>0$, there is a clause length $k$ such that $k$-SAT on $n$ variables cannot be solved in $2^{(1-\varepsilon)n}$ time. Due to the Sparsification Lemma \cite{Impagliazzo2001SETH}, the formulas that \SETH ~concerns have {\em linear} size.
It is much easier to believe that $\CircuitSAT$ for an unrestricted circuit (as opposed to a formula), of much larger, $2^{o(n)}$ size requires enumeration of all possible inputs.
\end{rem}
		
\begin{conjecture} \label{MinPlus}
	Online $(\min,+)$ product between an integer $n \times n$ matrix and $n$ length-$n$ vectors requires $n^{3 - o(1)}$ time in the word-RAM model of computation with $O(\log n)$ bit words. 
	That is, given an integer matrix $A \in \mathbb{Z}^{n \times n}$, and $n$ vectors $v^1,v^2,\dotsc,v^n$ that are revealed one by one, we wish to compute the $(\min,+)$-products 
	$$
	A \diamond v ~~:=~~ \left(\min_{k=1}^{n} (A_{1,k} + v_k) , ~~ \min_{k=1}^{n} (A_{2,k} + v_k),~~\dotsc,~~\min_{k=1}^{n} (A_{n,k} + v_k)\right)
	$$
	between $A$ and each of the $v^i$s. We get to access $v^{i+1}$ only after we have output $A \diamond v^{i}$.
    The conjecture asserts that the whole computation requires $n^{3 - o(1)}$ time.
\end{conjecture}
        
\begin{rem}
	The \emph{offline} (and thus easier) version of the above problem is equivalent to calculating the $(\min,+)$-product of two matrices of size $n \times n$, which is known to be asymptotically equivalent to the famous $\mathsf{APSP}$ problem \cite{fischermeyer}: $(\min,+)$-product is in $O(n^c)$ time if and only if $\mathsf{APSP}$ is in $O(n^c)$ time, for any constant $c$.

The online $(\min,+)$-product conjecture is a natural generalization of the online Boolean Matrix-Vector Product conjecture of Henzinger et al. \cite{henzinger2015unifying} that asserts that given a Boolean $n\times n$ matrix, multiplying it with $n$ Boolean vectors given online requires $n^{3-o(1)}$ time, in the Word-RAM model. There is no known relationship between the APSP conjecture and the Online Boolean Matrix-Vector Product conjecture, so one may be true even if the other fails. It is not hard to embed Boolean product into $(\min,+)$-product, and hence our conjecture is a weakening of both of these conjectures simultaneously, making ours very believable.
\end{rem}
	
\begin{conjecture}[3-SUM Conjecture] \label{3SUM}
There exists a constant $q$, so that
given three size-$n$ sets $A$, $B$, $C$ of integers in $[-n^{q},n^{q}]$, deciding whether there exist $a \in A$, $b \in B$, $c \in C$ such that $a + b + c = 0$ requires $n^{2 - o(1)}$ time in the word-RAM model with $O(\log n)$ bit words.
\end{conjecture}

\begin{rem}
    The 3-SUM Conjecture was the first attempt to address fine-grained complexity, back in 1995 \cite{Gajentaan-Overmars-1995}.
    By a standard hashing trick, we can assume $q\le 3+\delta$ for any $\delta>0.3$~\cite{VassilevskaW2018survey}.
    It remains open despite several slightly subquadratic algorithms \cite{Baran-Demaine-Patrascu-2007,chan18,Gronlund-Pettie-2014}.
\end{rem}
		
    We can now state our lower bounds conditioned on the conjectures above, whose proofs are in Section~\ref{LowerBounds}.
	As in our conjectures above, throughout the paper, we assume that we are working in the word-RAM model with word size $w = \Theta(\log \max\{n,m\})$, where $n$ denotes the size of the data structure problem and $m$ denotes the length of the update sequence (timeline).
 
	\begin{theo}\label{theo:SETH-based}
		There is a data structure problem that has an $O(n^{1+o(1)})$-time partially retroactive data structure, but conditioned on Conjecture~\ref{CircuitSAT}, requires $\Omega(n^{2-o(1)})$ time for fully retroactive queries when $m = \Theta(n^2)$.
	\end{theo}

	\begin{theo}\label{theo:minplus-based}
		There is a data structure problem that has an $O(\log n)$-time partially retroactive data structure, but conditioned on Conjecture~\ref{MinPlus}, requires $\Omega(n^{1 - o(1)})$ time for fully retroactive queries when $m = \Theta(n^2)$.
	\end{theo}

	\begin{theo}\label{theo:3SUM-based}
		There is a data structure problem that has an $O(\sqrt n)$-time partially retroactive data structure, but conditioned on Conjecture~\ref{3SUM}, requires $\Omega(n^{1-o(1)})$ time for fully retroactive queries when $m = \Theta(n)$.
	\end{theo}

\paragraph*{Our Results: Matching Upper bound}

	The three theorems above show that improving the general dependence on $\sqrt{m}$ is impossible based on any of these three conjectures. But we may hope to have a better data structure when $m \gg n^2$. In fact, we show in Section~\ref{UpperBound} that this is possible, for any ``reasonable'' data structure, by establishing the following theorem:
	
	\begin{theo}\label{theo:DS}
		Suppose a data structure of size $n$ satisfies the following conditions:
		\begin{enumerate}
		\item There is a sequence of $O(n)$ queries to extract the whole state\footnote{The state of a data structure is a description of all data it currently stores.} $\mathcal{S}$ from it.
		\item Given a state $\mathcal{S}$ of size $n$, there is a sequence of $O(n)$ operations to update the data structure from empty initial state to $\mathcal{S}$.
		\item It is partially retroactive with operation time $T_{\mathsf{op}}(n,m)$.
		\end{enumerate}
Then the corresponding problem has an amortized fully retroactive data structure with operation time\\$O\left( \min\{\sqrt{m},n\log m\} \cdot T_{\mathsf{op}}(n,m) \right)$.
	\end{theo}
    
	\begin{rem}
		The data structure of Theorem~\ref{theo:DS} is similar to the data structure described in \cite[Section~2.2]{demaine2015polylogarithmic}.
	\end{rem}

	Combining the above four theorems, we conclude that under reasonable conditions, the optimal gap between partial and full retroactivity
    is $\Theta(\min\{\sqrt{m},n\})$, up to $m^{o(1)}$ factors, for any $n$ and $m$.
	
	\paragraph*{Related Work}
	The field of fine-grained complexity studies the exact running time for problems in $\PTIME$ and beyond, and proves many lower bounds for data structure problems conditioned on various conjectures~\cite{Patrascu10,abboud2014popular,henzinger2015unifying,kopelowitz2016higher,abboud2016popular,henzinger2017conditional,Goldstein2017ST}. Look at the recent survey \cite{VassilevskaW2018survey} for a summary of the known results in fine-grained complexity. We mention two of the related papers. 
Building on work by Patrascu~\cite{Patrascu10} who focused on the $3$-SUM conjecture,  Abboud and Vassilevska W.~\cite{abboud2014popular} proved hardness for data structure problems under a variety of hypotheses: $\SETH$, $3$-SUM, APSP etc. \cite{abboud2014popular} introduced $\SETH$ as a hardness hypothesis for data structure problems and 
obtained $\SETH$-hardness for the following dynamic problems: maintaining under edge updates (insertions or deletions) the strongly connected components of a graph, the number of nodes reachable from a fixed source, a $1.3$-approximation of the diameter of the graph, or whether there is $(s,t) \in S \times T$ such that $s$ can reach $t$ for two fixed node sets $S$ and $T$. 
    Henzinger et al.~\cite{henzinger2015unifying} introduces the Online Matrix-Vector Multiplication Conjecture, and shows that it implies tight hardness result for subgraph connectivity, Pagh’s problem, $d$-failure connectivity, decremental single-source shortest paths, and decremental transitive closure.
	
\section{Lower Bounds}
\label{LowerBounds}
	
	In this section, we first give a data structure framework, which eases the construction of our separation, and then we prove Theorem~\ref{theo:SETH-based}, Theorem~\ref{theo:minplus-based}, and Theorem~\ref{theo:3SUM-based}.
	
	\subsection{Data Structure Framework}
	
	We present a data structure framework which turns out to be easy for partially retroactive data structures, but hard for their fully retroactive counterparts. In this framework, a data structure $\mathcal{D}$ maintains several lists, and answers a certain question on them. The formal definition is given below.
	
	\begin{fram}[Data Structure Problem $\mathcal{P}_F$]
		In our data structure problem $\mathcal{P}_F$. We are required to maintain a constant number of lists consisting of items from an entry set $\entry$.
        Denote the lists as $\mathcal{L}_1,\mathcal{L}_2,\dotsc,\mathcal{L}_{k}$, and $F$ is a function defined on these lists. 
		
		We can view each list $\mathcal{L}_i$ as a mapping from $\mathbb{N}$ to $\entry$ and initially every list maps all indices to the idle symbol $\perp$. We use $\mathcal{L}[a]$ to denote the $a$-th element of the list $\mathcal{L}$, and we measure the size of a list $\mathcal{L}$ (denoted by $|\mathcal{L}|$) by the number of $a$'s such that $\mathcal{L}[a] \ne \perp$. The size of the data structure is then measured by sum of the sizes of all its lists.\footnote{A list can also be viewed as a dictionary over integers. We view them as lists because, in our construction, it is much more convenient to do so.}
		
		There are two types of operations.
		
		\begin{itemize}
			\item $\textit{set-element}(\mathcal{L}_i,a,e)$: Set $\mathcal{L}_i[a] = e$.
			
			\item $\textit{$F$-evaluation}$: Evaluate $F$ on the current maintained lists $\mathcal{L}_1,\mathcal{L}_2,\dotsc,\mathcal{L}_k$.
		\end{itemize}
					
	\end{fram}

	The key property for the problem $\mathcal{P}_F$ is that, once we have a data structure $\mathcal{D}_F$ for it, it supports partially retroactive queries with essentially no overhead.
	
	\begin{lemma}\label{lm:easy-partial}
		Suppose there is a data structure $\mathcal{D}_F$ for the data structure problem $\mathcal{P}_F$ with update time $T_U$ and query time $T_Q$. Then there is a partially retroactive data structure $\mathcal{D}_F^{\partver}$ for problem $\mathcal{P}_F$ with update time $T_U + O(\log m)$ and query time $T_Q$.
	\end{lemma}
	\begin{proof}
		Our partially retroactive data structure $\mathcal{D}_F^\partver$ simply simulates an instance of the regular structure $\mathcal{D}_F$ which represents the current version of the data structure. Whenever there is an update in the history, it could be either inserting or deleting an operation $\textit{set-element}(\mathcal{L}_i,a,e)$ at time $t$, it only affects the $a$-th element in $\mathcal{L}_i$ of the current version of the data structure $\mathcal{D}_F$. 
		
		Therefore, we can use a BST to organize all $\textit{set-element}$ operations on each location of each list in the chronological order. We update the corresponding BST on the $a$-th element of list $\mathcal{L}_i$ when dealing with insertion or deletion of an operation $\textit{set-element}(\mathcal{L}_i,a,e)$ in the history. When the latest $\textit{set-element}$ changes in the BST (or the BST becomes empty), we update the corresponding value in $\mathcal{D}_F$. And the query operation is equivalent to the same query operation on the current data structure $\mathcal{D}_F$. The time cost is the usual time cost of BST.
	\end{proof}

	\subsection{\boldmath Lower Bound from $\SIZE(2^{o(n)})$ $\CircuitSAT$}
	
	\newcommand{\WL}{\widetilde{\mathcal{L}}}
	\newcommand{\CL}{\mathcal{L}}
	
	Now we are ready to prove our lower bounds. First we prove Theorem~\ref{theo:SETH-based}, which we repeat here for completeness:
	
\begin{reminder}{Conjecture~\ref{CircuitSAT}}
  In the Word-RAM model with $O(\log n)$ bit words, it takes $2^{n - o(n)}$ time to solve  $\CircuitSAT$ on $n$-input circuits of size $2^{o(n)}$.
\end{reminder}
	
	\begin{reminder}{Theorem~\ref{theo:SETH-based}}
		There is a data structure problem that has an $O(n^{1+o(1)})$-time partially retroactive data structure, but conditioned on Conjecture~\ref{CircuitSAT}, requires $\Omega(n^{2 - o(1)})$ time for fully retroactive queries when $m = \Theta(n^2)$.
	\end{reminder}
	
	\begin{proof}
		Let $d=n^{o(1)}$.
        We use the entry set 
	 	$$
	 	\entry := \mathcal{C}_{d} \times \{0,1\}^{\le d},
	 	$$
	 	where $\mathcal{C}_{d}$ is the set of descriptions of all circuits of size at most $d$, and $\{0,1\}^{\le d}$ is the set of binary strings of length at most $d$.
	 	These descriptions take at most $O(\poly(d)) = n^{o(1)}$ bits. Therefore, an item from $\entry$ consists of $n^{o(1)}$ bits. Denote this number by $d^\prime$. 
		
		Consider the data structure problem $\mathcal{P}_{F^{(\SAT)}}$ with respect to two lists $\mathcal{L}_1,\mathcal{L}_2$ of items in $\entry$ and the function $F^{(\SAT)}$ defined on them as follows. $F^{(\SAT)}(\mathcal{L}_1,\mathcal{L}_2) = 1$ if the following holds:
		
		\begin{quote}\begin{quote}
		There exist $a$ and $b$ with $\mathcal{L}_1[a] = (C_1,x_1) \ne \perp$ and $\mathcal{L}_2[b] = (C_2,x_2) \ne \perp$ such that
			\begin{itemize}
			\item $C_1 = C_2$;
            \item $C_2$ is a valid description of a circuit of size at most $d$ with exactly $|x_1| + |x_2|$ bits of input;
			\item $C_2(x_1,x_2) = 1$.
			\end{itemize}
		\end{quote}\end{quote}
		$F^{(\SAT)}(\mathcal{L}_1,\mathcal{L}_2) = 0$ otherwise. We say a pair $(C_1,x_1)$ and $(C_2,x_2)$ is \emph{good} if they satisfy the conditions above.
		 
		\newcommand{\Nmin}{N_{\mathsf{min}}}
		
		Let $\ell := \left(|\mathcal{L}_1| + |\mathcal{L}_2| \right)$.
        The size of the whole structure is $n = d^\prime \ell$.
		
		\highlight{Partially Retroactive Upper Bound.} In order to maintain $F^{(\SAT)}(\CL_1,\CL_2)$, we keep a counter $n_{\SAT}$ recording the number of pairs $a$ and $b$ such that $\CL_1[a]$ and $\CL_2[b]$ is a good pair. Whenever we modify an element in lists $\CL_1$ or $\CL_2$, it takes $O(n^{1 + o(1)})$ time to update the counter $n_{\SAT}$.
		
		Now, since we have an $O(n^{1 + o(1)})$ update time algorithm for $\mathcal{P}_{F^{(\SAT)}}$, by Lemma~\ref{lm:easy-partial}, it extends to an algorithm for the partially retroactive version.
		
		\highlight{Fully Retroactive Lower Bound.} Given a circuit $C$ of size $2^{o(u)}$ with $u$ inputs. Let $\ell = 2^{u/4}$ be the size of the lists in the data structure (assuming $u$ is divisible by $4$ for simplicity).
		
		Let $A$ and $B$ be two identical lists of entries in $\entry$ with size $2^{u/2} = \ell^2$, such that the $i$-th element of $A$ and $B$ is $(C,w_i)$, where $w_i$ is the $i$-th length $u/2$ binary string in lexicographic order. Then we divide $A$ and $B$ into $\ell = 2^{u/4}$ groups of equal size, and denote them by $A_1,A_2,\dotsc,A_{\ell}$ and $B_1,B_2,\dotsc,B_{\ell}$ correspondingly, where each $A_i$ and each $B_i$ is a list of size $\ell$.
		
		The circuit $C$ is satisfiable if and only if there exists $a \in A$ and $b \in B$ such that $a$ and $b$ is a good pair. Consider the following operation sequences:
		
		\begin{itemize}
			\item First, for each $k \in [\ell]$, we add an operation $\textit{set-element}(\mathcal{L}_1,k,\perp)$. We denote the operation time by $t_{k}$.
			\item Next for each $j \in [\ell]$, we add an operation $\textit{set-element}(\mathcal{L}_{2},k,B_{j}[k])$ for each $k \in [\ell]$. We denote the time right after adding the last operation for each $j$ ($\textit{set-element}(\mathcal{L}_2,\ell,B_{j}[\ell])$) by $q_j$.
			\item Now, for each $i \in \left[\ell\right]$, we replace the operation on time $t_k$ by an operation $\textit{set-element}(\mathcal{L}_1,k,A_{i}[k])$ for each $k \in \left[\ell\right]$, and after that, we make fully retroactive query $F^{(\SAT)}$-$\textit{evaluation}$ at time $q_j$ for each $j \in \left[\ell \right]$. From the definition of $F^{(\SAT)}$, it tells us whether there exists $a \in A_i$, $b \in B_j$ such that $a$ and $b$ is a good pair, for each $i,j \in \left[\ell\right]$.
		\end{itemize}
		
		The whole procedure consists of $m = \Theta(\ell^2)=O(n^2)$ operations. Conditioning on Conjecture~\ref{CircuitSAT}, the whole sequence takes at least $2^{u(1-o(1))} = \ell^{4-o(1)}=n^{4-o(1)}$ time, which means a fully retroactive operation takes at least amortized $\Omega(n^{2-o(1)})$ time, and completes the proof.
	\end{proof}

	\subsection{\boldmath Lower Bounds from Online $(\min,+)$-product}
	
    Next we prove Theorem~\ref{theo:minplus-based}, which we recap here for completeness:
    
\begin{reminder}{Conjecture~\ref{MinPlus}}
	Online $(\min,+)$ product between an integer $n \times n$ matrix and $n$ length-$n$ vectors requires $n^{3 - o(1)}$ time in the word-RAM model with $O(\log n)$ bit words. 
	That is, given an integer matrix $A \in \mathbb{Z}^{n \times n}$, and $n$ vectors $v^1,v^2,\dotsc,v^n$ which are revealed one by one, we wish to compute the $(\min,+)$-product 
	$$
	A \diamond v ~~:=~~ \left(\min_{k=1}^{n} (A_{1,k} + v_k) , ~~ \min_{k=1}^{n} (A_{2,k} + v_k),~~\dotsc,~~\min_{k=1}^{n} (A_{n,k} + v_k)\right)
	$$
	between $A$ and each of the $v^i$s. We get to access $v^{i+1}$ only after we have output $A \diamond v^{i}$.
    The conjecture asserts that the whole computation requires $n^{3 - o(1)}$ time.
\end{reminder}
	
	\begin{reminder}{Theorem~\ref{theo:minplus-based}}
		There is a data structure problem that has an $O(\log n)$-time partially retroactive data structure, but conditioned on Conjecture~\ref{MinPlus}, requires $\Omega(n^{1 - o(1)})$ time for fully retroactive queries when $m = \Theta(n^2)$.
	\end{reminder}

	\begin{proof}
		Let $c$ be a constant such that all entries from $A$ and all $v^i$'s lie in $[0,n^c]$.
		
		Now, consider the data structure problem $\mathcal{P}_{F^{(\min,+)}}$ with respect to two lists $\mathcal{L}_1,\mathcal{L}_2$ and the function $F^{(\min,+)}$ defined on them as
		$$
		F^{(\min,+)}(\mathcal{L}_1,\mathcal{L}_2) := \min_{ a : \mathcal{L}_1[a] \ne \perp, \mathcal{L}_2[a] \ne \perp} \left( \mathcal{L}_1[a] + \mathcal{L}_2[a] \right).
		$$
		The entry set $\entry$ here is the integers in $[0,n^{c}]$.

		\highlight{Partially Retroactive Upper Bound.} Clearly, the operations in $\mathcal{P}_{F^{(\min,+)}}$ can be supported in $O(\polylog(n))$ time: we use a priority queue to maintain the sums $\mathcal{L}_1[a] + \mathcal{L}_2[a]$ for all the valid $a$'s, and update the priority queue correspondingly after each $\textit{set-element}$ operations. Therefore, by Lemma~\ref{lm:easy-partial}, we know the update/query operations in the partially retroactive version of $\mathcal{P}_{F^{(\min,+)}}$ can be supported in $O(\polylog(n) + \log m)$ time.
		
		\highlight{Fully Retroactive Lower Bound.} Let $a_1,a_2,\dotsc,a_n$ be the $n$ rows of $A$, and $v$ be a vector. Computing the $(\min,+)$ product of $A$ and $v$ is equivalent to compute
		$$
		(a_i \diamond v) := \min_{k=1}^{n} \left( a_{i,k} + v_{k} \right)
		$$
		for each $i \in [n]$.
		
		We are going to show that a fully retroactive algorithm for $\mathcal{P}_{F^{(\min,+)}}$ can be utilized to compute $(a_i \diamond v^j)$ for each $i,j \in [n]$ in an online fashion.
		
		Consider the following operation sequences. First we add $\textit{set-element}(\mathcal{L}_1,k,0)$ for each $k \in [n]$; then for each $j \in [n]$, we add $\textit{set-element}(\mathcal{L}_2,k,a_{j,k})$ for each $k \in [n]$. We use $t_{j}$ to denote the time right after adding the operation $\textit{set-element}(\mathcal{L}_2,n,a_{j,n})$, i.e., the time we have just set $\mathcal{L}_2$ to represent vector $a_j$.
		
		Then for each $i \in [n]$, we delete the first $n$ operations in the history (that is, we clear all the $\textit{set-element}$ operations on $\mathcal{L}_1$); and then we add $\textit{set-element}(\mathcal{L}_1,k,v^i_{k})$ for each $k \in [n]$ at the beginning of the operation sequence (that is, we set $\mathcal{L}_1$ to represent the vector $v^i$); next we make a fully retroactive query $\textit{$F^{(\min,+)}$-evaluation}$ at the time $t_{j}$ for each $j \in [n]$.
        It is easy to see that querying at time $t_j$ gives us the value of $(a_j \diamond v^i)$. So, after performing the above procedure for $v^i$, we have calculated the $(\min,+)$ product between $A$ and $v^i$.
		
		The size of data structure is $\Theta(n)$, and there are $m = \Theta(n^2)$ operations in total. Hence, conditioned on Conjecture~\ref{MinPlus}, any fully retroactive data structure running on the above algorithm takes at least amortized $n^{1-o(1)}$ time for either update or query operation.
	\end{proof}

	\subsection{Lower Bounds from 3-SUM}
	
	Next, we prove Theorem~\ref{theo:3SUM-based}, which we recap here for completeness:
	
\begin{reminder}{Conjecture \ref{3SUM} (3-SUM Conjecture)}  There is a constant $q$ such that, given three size-$n$ sets $A$, $B$, $C$ of integers in $[-n^{q},n^{q}]$, deciding whether there exist $a \in A$, $b \in B$, $c \in C$ such that $a + b + c = 0$ requires $n^{2 - o(1)}$ time in the word-RAM model with $O(\log n)$ bit words.
\end{reminder}

	\begin{reminder}{Theorem~\ref{theo:3SUM-based}}
		There is a data structure problem that has an $O(\sqrt n)$-time partially retroactive data structure, but conditioned on Conjecture~\ref{3SUM}, requires $\Omega(n^{1-o(1)})$ time for fully retroactive queries when $m = \Theta(n)$.
	\end{reminder}
    
	\begin{proof}
		Consider the data structure problem $\mathcal{P}_{F^{(\text{3SUM})}}$ with respect to three lists $\mathcal{L}_1,\mathcal{L}_2,\mathcal{L}_3$ and the function $F^{(\text{3SUM})}$ defined on them as follows
		$$
		F^{(\text{3SUM})}(\mathcal{L}_1,\mathcal{L}_2,\mathcal{L}_3) :=
		\begin{cases}
			1 \quad &\text{$|\mathcal{L}_2|^2 \le |\mathcal{L}_1|$, $|\mathcal{L}_3|^2 \le |\mathcal{L}_1|$, and there exist $a,b,c$ such that}\\
					&\text{$\mathcal{L}_1[a] \ne \perp$, $\mathcal{L}_2[b] \ne \perp$, $\mathcal{L}_3[c] \ne \perp$ and $\mathcal{L}_1[a] + \mathcal{L}_2[b] + \mathcal{L}_3[c] = 0$;}\\
			0 \quad &\text{otherwise}.       
		\end{cases}
		$$
		
		Let $n := \sum_{i=1}^{3} |\mathcal{L}_i|$ be size of the whole structure, and $n_i := |\mathcal{L}_i|$.
		
		\highlight{Partially Retroactive Upper Bound.} We use $\WL_2$ (resp.~$\WL_3$) to denote the sublists consisting of the first (at most) $\sqrt{n_1}$ elements of $\mathcal{L}_2$ (resp.~$\mathcal{L}_3$). Then by maintaining a BST for each list, an operation on $\mathcal{L}_2$ (resp.~$\mathcal{L}_3$) can be easily reduced to at most one operation on $\WL_2$ (resp.~$\WL_3$). Since whenever $\WL_2 \ne \mathcal{L}_2$ or $\WL_3 \ne \mathcal{L}_3$, $F(\mathcal{L}_1,\mathcal{L}_2,\mathcal{L}_3)$ is defined to be zero, we can pretend to work with $\WL_2$ and $\WL_3$. 
		
		We build a hash table $\mathcal{H}$ storing all the elements in $\mathcal{L}_1$, and every value of the form $- a - b$ for $a \in \WL_2$, $b \in \WL_3$. Using this table, we can count and maintain the number of the triples $(a,b,c)$ such that $\mathcal{L}_1[a] + \WL_2[b] + \WL_3[c] = 0$. We denote this number by $n_{\text{triple}}$.
		
		Whenever we modify the list $\mathcal{L}_1$, we make the corresponding change on $\mathcal{H}$. This may also cause $O(1)$ additional operations on $\WL_2$ and $\WL_3$, as $n_1$ can be larger or smaller. And when we modify the list $\WL_2$ or $\WL_3$, this causes updating at most $\max(|\WL_2|,|\WL_3| = O(\sqrt{n})$ values in $\mathcal{H}$.
		
		Since we have an $O(\sqrt{n})$ update time algorithm for $\mathcal{P}_{F^{(\text{3SUM})}}$, by Lemma~\ref{lm:easy-partial}, it extends to an algorithm for the partially retroactive version.
		
		\highlight{Fully Retroactive Lower Bound.} 
        Let $A$, $B$, $C$ be three integer lists of size $n$. 
        For convenience we assume that $n$ is a square number. We divide $B$ and $C$ into $\sqrt{n}$ groups of equal size, and denote them by $B_1,B_2,\dotsc,B_{\sqrt{n}}$ and $C_1,C_2,\dotsc,C_{\sqrt{n}}$ correspondingly. Then each $B_i$ and each $C_i$ is a list of size $\sqrt{n}$.
		
		Consider the following operation sequence. 
		
		\begin{itemize}
			\item First, for each $i \in [n]$, we add an operation $\textit{set-element}(\mathcal{L}_1,i,A[i])$, that is, we set the list $\mathcal{L}_1$ to represent the set $A$; then for each $k \in [\sqrt{n}]$, we add an operation $\textit{set-element}(\mathcal{L}_2,k,0)$, whose operation time is denoted by $t_{k}$.
			\item Next for each $j \in [\sqrt{n}]$, we add an operation $\textit{set-element}(\mathcal{L}_{3},k,C_{j}[k])$ for each $k \in [\sqrt{n}]$. We denote the time right after adding the operation $\textit{set-element}(\mathcal{L}_3,\sqrt{n},C_{j}[\sqrt{n}])$ as time $q_j$.
			
			\item Now, for each $i \in \left[\sqrt{n}\right]$, we replace the operation on time $t_k$ by an operation $\textit{set-element}(\mathcal{L}_2,k,B_{i}[k])$. After that, we make a fully retroactive query $F^{\text{3SUM}}$-$\textit{evaluation}$ at time $q_j$ for each $j \in \left[\sqrt{n} \right]$. From the definition of $F^{\textit{3SUM}}$, the queries tell us whether there exists $a \in A$, $b \in B_i$, $c \in C_j$ such that $a + b + c = 0$ for each $i,j \in \left[\sqrt{n}\right]$, and thus solve the 3SUM problem.
		\end{itemize}
		
		The data structure above has size $\Theta(n)$, and the whole procedure consists of $ m = \Theta(n)$ operations. Therefore, conditioned on Conjecture~\ref{3SUM}, either update or query for a fully retroactive data structure for problem $\mathcal{P}_{F^{(\text{3SUM})}}$ takes amortized $\Omega(n^{1-o(1)})$ time.
	\end{proof}
	
	\section{Upper Bounds} \label{UpperBound}
	
	In this section, we prove Theorem~\ref{theo:DS}:
	
	\begin{reminder}{Theorem~\ref{theo:DS}}		
		Suppose a data structure of size $n$ satisfies the following conditions:
		\begin{enumerate}
		\item There is a sequence of $O(n)$ queries to extract the whole state $\mathcal{S}$ from it.
		\item Given a state $\mathcal{S}$ of size $n$, there is a sequence of $O(n)$ operations to update the data structure from empty initial state to $\mathcal{S}$.
		\item It is partially retroactive with operation time $T_{\mathsf{op}}(n,m)$.
		\end{enumerate}
		Then the corresponding problem has an amortized fully retroactive data structure with operation time\\$O\left( \min\{\sqrt{m},n\log m\} \cdot T_{\mathsf{op}}(n,m) \right)$.
	\end{reminder}

	\begin{proof}
		We use a weight-balanced binary tree (WBT) $\mathcal{T}$ to maintain the whole operation sequence \cite{galperin1993scapegoat}. The subtree of each node $u$ corresponds to an interval of operations $S_{u}$ in the whole operation sequence. We can build a partially retroactive data structure $\mathcal{D}_u$ on $S_{u}$ as augmented information in node $u$. One property of WBT is that when we insert or delete its nodes, the amortized total number of element changes to all $S_{u}$ is only $O(\log m)$. More formally, if $S_{u}$ is the set of operations before a node insertion or deletion, and $S_{u}^\prime$ is the set of operations after the insertion or deletion, then WBT ensures 
		$$\sum_{u} |S_{u} \setminus S_{u}^\prime| + |S_{u}^\prime \setminus S_{u}|$$
		is  amortized $O(\log m)$. For each element change in $S_u$, we can update $\mathcal{D}_u$ using the partially retroactive data structure in $O(T_{\mathsf{op}}(n,m) \cdot \log m)$ amortized time per insert/delete of an operation.
		
		For each fully retroactive query, we first extract the corresponding prefix of the operation sequence from the WBT. By properties of WBT, in $O(\log m)$ time, we can get $k = O(\log m)$ nodes, $u_1,u_2,\dotsc,u_k$, such that the concatenation of these $S_{u_i}$'s is exactly the prefix we are asking. Next we maintain a data structure state $\mathcal{S}$ initialized as the empty state. We go through each $u_i$ in order: first append $O(n)$ operations at the beginning of $\mathcal{D}_u$ to set the initial state inside $\mathcal{D}_u$ to be $\mathcal{S}$, and then make $O(n)$ queries on $\mathcal{D}_u$, to extract its final state, and set $\mathcal{S}$ to be that state. By a simple induction, we can see that after we finished processing node $u_i$, the final state of $\mathcal{D}_{u_i}$ corresponds to the state resulting from the concatenation of $S_{u_1},S_{u_2},\dotsc,S_{u_i}$. Therefore, we can then query $\mathcal{D}_{u_k}$ to get the answer we want. Finally, we delete all the operations we added in those $\mathcal{D}_{u}$, so they can be used for the future queries. To summarize, we invoke partially retroactive update/query $O(n \cdot \log m)$ times, and hence the whole query takes $O(n \cdot \log m \cdot T_{\mathsf{op}}(n,m))$ time.
		
		Demaine et al.~\cite{demaine2004retroactive} showed a reduction with $O(\sqrt{m})$ overhead. Roughly, their transformation maintains $\sqrt{m}$ equally distributed checkpoints, and for each checkpoint, they maintain a partially retroactive data structure for the prefix up to that checkpoint. For update, they need to update all the $\sqrt{m}$ partially retroactive data structures; for query of a prefix, they first find the closest checkpoint, adding or deleting operations to this checkpoint in order for it to match the prefix, and then do the query. For both update and query, there are $O(\sqrt m)$ calls to the partially retroactive data structure, hence the $O(\sqrt m)$ overhead. 
		
		Combining these two transformations gives an $O(\min\{\sqrt{m}, n \cdot \log m\})$ overhead. There is a subtle issue here as this requires us to know $n$ and $m$ beforehand. We can avoid that by using the standard technique that maintains two structures $\mathcal{D}_1$ and $\mathcal{D}_2$ simultaneously, one with $\sqrt{m}$ overhead and one with $n \cdot \log m$ overhead. We simulate $\mathcal{D}_1$ and $\mathcal{D}_2$ in an interleaving fashion, and answer the query as soon as one of them gives its answer.
	\end{proof}
    
\section{Discussion}
		Many lower bounds for algorithm problems are based on plausible conjectures from fine-grained complexity theory. Besides the three canonical ones (\SETH, \textsf{APSP}, 3-SUM) mentioned above, some interesting hardness candidates include Boolean Matrix Multiplication~\cite{VassilevskaW10equiv}, Online Matrix Vector Multiplication~\cite{henzinger2015unifying}, and the Triangle Collection problem~\cite{Abboud2015tria}. Their relationship and applications are discussed in detail in~\cite{VassilevskaW2018survey}. 

		Our lower bound constructions reveal that fully retroactive queries facilitate batched pair evaluation. We believe this technique can prove useful for other data structure lower bounds, especially dynamic ones. Some examples include the total update time for partially-dynamic algorithms, worst-case update time, query/update time tradeoffs~\cite{henzinger2015unifying}, and space/time tradeoffs~\cite{Goldstein2017ST}.
		
\section*{Acknowledgment}

We would like to thank Quanquan Liu and Ryan Williams for helpful discussions, and the anonymous reviewers for their generous comments.
	
	\bibliography{team}

\begin{thebibliography}{10}

\bibitem{abboud2016popular}
Amir Abboud and S{\o}ren Dahlgaard.
\newblock Popular conjectures as a barrier for dynamic planar graph algorithms.
\newblock In {\em Proceedings of the IEEE 57th Annual Symposium on Foundations
  of Computer Science}, pages 477--486, 2016.

\bibitem{Abboud2015tria}
Amir Abboud, Virginia Vassilevska~Williams, and Huacheng Yu.
\newblock Matching triangles and basing hardness on an extremely popular
  conjecture.
\newblock In {\em Proceedings of the Forty-seventh Annual ACM Symposium on
  Theory of Computing}, STOC '15, pages 41--50, New York, NY, USA, 2015. ACM.

\bibitem{abboud2014popular}
Amir Abboud and Virginia~Vassilevska Williams.
\newblock Popular conjectures imply strong lower bounds for dynamic problems.
\newblock In {\em Proceedings of the IEEE 55th Annual Symposium on Foundations
  of Computer Science}, pages 434--443, 2014.

\bibitem{Baran-Demaine-Patrascu-2007}
Ilya Baran, Erik~D. Demaine, and Mihai P\v{a}tra\c{s}cu.
\newblock Subquadratic algorithms for {3SUM}.
\newblock {\em Algorithmica}, 50(4):584--596.

\bibitem{Blelloch-2008}
Guy~E. Blelloch.
\newblock Space-efficient dynamic orthogonal point location, segment
  intersection, and range reporting.
\newblock In {\em Proceedings of the 19th Annual ACM-SIAM Symposium on Discrete
  Algorithms}, pages 894--903, 2008.

\bibitem{chan18}
Timothy~M. Chan.
\newblock More logarithmic-factor speedups for 3{SUM},
  (median,$+$)-convolution, and some geometric 3sum-hard problems.
\newblock In {\em Proceedings of SODA 2018}, 2018.
\newblock to appear.

\bibitem{demaine2004retroactive}
Erik~D. Demaine, John Iacono, and Stefan Langerman.
\newblock Retroactive data structures.
\newblock In {\em Proceedings of the 15th Annual ACM-SIAM Symposium on Discrete
  Algorithms}, pages 274--283, 2004.

\bibitem{demaine2007retroactive}
Erik~D. Demaine, John Iacono, and Stefan Langerman.
\newblock Retroactive data structures.
\newblock {\em ACM Transactions on Algorithms}, 3(2):13:1--13:21, 2007.

\bibitem{demaine2015polylogarithmic}
Erik~D. Demaine, Tim Kaler, Quanquan Liu, Aaron Sidford, and Adam Yedidia.
\newblock Polylogarithmic fully retroactive priority queues via hierarchical
  checkpointing.
\newblock In {\em Proceedings of the 14th International Symposium on Workshop
  on Algorithms and Data Structures}, pages 263--275. Springer, 2015.

\bibitem{Dickerson-Eppstein-Goodrich-2010}
Matthew~T. Dickerson, David Eppstein, and Michael~T. Goodrich.
\newblock Cloning voronoi diagrams via retroactive data structures.
\newblock In {\em Proceedings of the 18th Annual European Symposium on
  Algorithms}, pages 362--373, 2010.

\bibitem{Driscoll-Sarnak-Sleator-Tarjan-1989}
James~R. Driscoll, Neil Sarnak, Daniel~D. Sleator, and Robert~E. Tarjan.
\newblock Making data structures persistent.
\newblock {\em Journal of Computer and System Sciences}, 38(1):86--124, 1989.

\bibitem{fischermeyer}
Michael~J. Fischer and Albert~R. Meyer.
\newblock Boolean matrix multiplication and transitive closure.
\newblock In {\em 12th Annual Symposium on Switching and Automata Theory, East
  Lansing, Michigan, USA, October 13-15, 1971}, pages 129--131, 1971.

\bibitem{Gajentaan-Overmars-1995}
Anka Gajentaan and Mark~H. Overmars.
\newblock On a class of {$O(n^2)$} problems in computational geometry.
\newblock {\em Computational Geometry: Theory and Applications}, 5:165--185,
  1995.

\bibitem{galperin1993scapegoat}
Igal Galperin and Ronald~L. Rivest.
\newblock Scapegoat trees.
\newblock In {\em Proceedings of the fourth annual ACM-SIAM Symposium on
  Discrete algorithms}, pages 165--174. Society for Industrial and Applied
  Mathematics, 1993.

\bibitem{Giora-Kaplan-2009}
Yoav Giora and Haim Kaplan.
\newblock Optimal dynamic vertical ray shooting in rectilinear planar
  subdivisions.
\newblock {\em ACM Trans. Algorithms}, 5(3):28:1--28:51, July 2009.

\bibitem{Goldstein2017ST}
Isaac Goldstein, Tsvi Kopelowitz, Moshe Lewenstein, and Ely Porat.
\newblock Conditional lower bounds for space/time tradeoffs.
\newblock In Faith Ellen, Antonina Kolokolova, and J{\"o}rg-R{\"u}diger Sack,
  editors, {\em Algorithms and Data Structures}, pages 421--436, Cham, 2017.
  Springer International Publishing.

\bibitem{Goodrich-Simons-2011}
Michael~T. Goodrich and Joseph~A. Simons.
\newblock Fully retroactive approximate range and nearest neighbor searching.
\newblock In {\em Proceedings of the 22nd International Symposium on Algorithms
  and Computation}, pages 292--301, 2011.

\bibitem{Gronlund-Pettie-2014}
Allan Gr{\o}nlund and Seth Pettie.
\newblock Threesomes, degenerates, and love triangles.
\newblock In {\em Proceedings of the IEEE 55th Annual Symposium on Foundations
  of Computer Science}, pages 621--630, October 2014.

\bibitem{henzinger2015unifying}
Monika Henzinger, Sebastian Krinninger, Danupon Nanongkai, and Thatchaphol
  Saranurak.
\newblock Unifying and strengthening hardness for dynamic problems via the
  online matrix-vector multiplication conjecture.
\newblock In {\em Proceedings of the 47th Annual ACM Symposium on Theory of
  Computing}, pages 21--30, 2015.

\bibitem{henzinger2017conditional}
Monika Henzinger, Andrea Lincoln, Stefan Neumann, and Virginia~Vassilevska
  Williams.
\newblock Conditional hardness for sensitivity problems.
\newblock {\em arXiv preprint arXiv:1703.01638}, 2017.

\bibitem{Impagliazzo2001SETH}
Russell Impagliazzo, Ramamohan Paturi, and Francis Zane.
\newblock Which problems have strongly exponential complexity?
\newblock {\em Journal of Computer and System Sciences}, 63(4):512--530, 2001.

\bibitem{kopelowitz2016higher}
Tsvi Kopelowitz, Seth Pettie, and Ely Porat.
\newblock Higher lower bounds from the 3sum conjecture.
\newblock In {\em Proceedings of the 27th Annual ACM-SIAM Symposium on Discrete
  Algorithms}, pages 1272--1287, 2016.

\bibitem{Nekrich-2010}
Yakov Nekrich.
\newblock Searching in dynamic catalogs on a tree.
\newblock arXiv:1007.3415, 2010.

\bibitem{Parsa-2014}
Salman Parsa.
\newblock {\em Algorithms for the {R}eeb Graph and Related Concepts}.
\newblock PhD thesis, Duke University, 2014.

\bibitem{Patrascu10}
Mihai Patrascu.
\newblock Towards polynomial lower bounds for dynamic problems.
\newblock In {\em Proceedings of the 42nd {ACM} Symposium on Theory of
  Computing, {STOC} 2010, Cambridge, Massachusetts, USA, 5-8 June 2010}, pages
  603--610, 2010.

\bibitem{VassilevskaW2018survey}
Virginia {Vassilevska Williams}.
\newblock On some fine-grained questions in algorithms and complexity.
\newblock In {\em Proceedings of the ICM}, 2018.
\newblock To appear.

\bibitem{VassilevskaW10equiv}
Virginia~Vassilevska Williams and Ryan Williams.
\newblock Subcubic equivalences between path, matrix and triangle problems.
\newblock In {\em Proceedings of the 2010 IEEE 51st Annual Symposium on
  Foundations of Computer Science}, FOCS '10, pages 645--654, Washington, DC,
  USA, 2010. IEEE Computer Society.

\end{thebibliography}
	\end{document}